\documentclass[letterpaper, 10pt, conference]{ieeeconf}  

\usepackage{graphicx,wrapfig}
\usepackage[utf8]{inputenc}
\usepackage[T1]{fontenc}
\usepackage{mathtools}
\usepackage[english]{babel}
\usepackage{comment}
\usepackage{amsmath}
\usepackage{amssymb}
\usepackage{cleveref}
\usepackage{amssymb}
\usepackage{bm}
\usepackage{psfrag}
\usepackage{esdiff}
\usepackage{float}
\usepackage{caption}
\usepackage{subfigure}
\usepackage{bm}
\usepackage{algorithmicx}
\usepackage{algpseudocode}
\usepackage{algorithm}
\usepackage{lipsum}
\usepackage{array}
\usepackage{blindtext}
\usepackage{epstopdf}
\usepackage{ragged2e}
\usepackage{mathrsfs}
\usepackage{tabu}
\graphicspath{ {figures/} }

\usepackage{url}
\usepackage{nomencl}
\usepackage{cite}
\usepackage{lmodern}
\usepackage{enumerate}
\usepackage{tikz}

\IEEEoverridecommandlockouts                              

\DeclarePairedDelimiter{\ceil}{\lceil}{\rceil}

\newcommand{\be}{\begin{equation}}
\newcommand{\ee}{\end{equation}}
\newcommand{\bse}{\begin{subequations}}
	\newcommand{\ese}{\end{subequations}}
\newcommand{\bewn}{\begin{equation*}}
\newcommand{\eewn}{\end{equation*}}
\newcommand{\bbmat}{\begin{bmatrix}} 
	\newcommand{\ebmat}{\end{bmatrix}}
\newcommand{\bd}{\begin{displaymath}}
\newcommand{\ed}{\end{displaymath}}
\newcommand{\bea}{\begin{eqnarray}}
\newcommand{\eea}{\end{eqnarray}}
\newcommand{\ba}{\begin{array}}
	\newcommand{\ea}{\end{array}}
\newcommand{\baa}{\begin{array}{ll}}
	\newcommand{\eaa}{\end{array}}

\newcommand{\bc}{\begin{center}}
	\newcommand{\ec}{\end{center}}
\newcommand{\ben}{\begin{enumerate}}
	\newcommand{\een}{\end{enumerate}}
\newcommand{\bi}{\begin{itemize}}
	\newcommand{\ei}{\end{itemize}}
\newcommand{\bt}{\begin{tabular}}
	\newcommand{\et}{\end{tabular}}
\newcommand{\bte}{\begin{table}}
	\newcommand{\ete}{\end{table}}

\newcommand{\norm}[1]{\left\lVert#1\right\rVert}   

\makeatletter
\renewcommand\paragraph{\@startsection{paragraph}{4}{\z@}%
	{-2.5ex\@plus -1ex \@minus -.25ex}%
	{1.25ex \@plus .25ex}%
	{\normalfont\normalsize\bfseries}}
\makeatother



\newtheorem{remark}{Remark}

\newtheorem{theorem}{Theorem}

\newtheorem{lemma}[theorem]{\textbf{Lemma}}
\newtheorem{assumption}{\textbf{Assumption}}
\newtheorem{problem}{\textbf{Problem}}
\newtheorem{definition}{\textbf{Definition}}

\newcommand{\bR}{\mathbb{R}}
\newcommand{\calA}{\mathcal{A}}
\newcommand{\calB}{\mathcal{B}}
\newcommand{\calC}{\mathcal{C}}
\newcommand{\calD}{\mathcal{D}}
\newcommand{\calE}{\mathcal{E}}
\newcommand{\calF}{\mathcal{F}}
\newcommand{\calG}{\mathcal{G}}
\newcommand{\calH}{\mathcal{H}}
\newcommand{\calM}{\mathcal{M}}
\newcommand{\calN}{\mathcal{N}}
\newcommand{\calO}{\mathcal{O}}
\newcommand{\calP}{\mathcal{P}}
\newcommand{\calQ}{\mathcal{Q}}
\newcommand{\calS}{\mathcal{S}}
\newcommand{\calSE}{\mathcal{SE}}
\newcommand{\calV}{\mathcal{V}}
\newcommand{\calW}{\mathcal{W}}
\newcommand{\calZ}{\mathcal{Z}}

\newcommand{\pdydx}[2]{\frac{\partial{#1}}{\partial{#2}}}

\newcommand{\abs}[1]{\left |#1\right |}

\title{\LARGE \bf
	Herding an Adversarial Swarm in an Obstacle Environment
}

\author{Vishnu S. Chipade and Dimitra Panagou
	\thanks{The authors are with the Department of Aerospace Engineering,
		University of Michigan, Ann Arbor, MI, USA;
		{\tt\small (vishnuc,dpanagou)@umich.edu}}
	\thanks{This work has been funded by the Center for Unmanned Aircraft Systems (C-UAS), a National Science Foundation Industry/University Cooperative Research Center (I/UCRC) under NSF Award No. 1738714 along with significant contributions from C-UAS industry members.}
}

\begin{document}
	\maketitle
	\thispagestyle{empty}
	\pagestyle{empty}
	
	\begin{abstract}
		
		This paper studies a defense approach against a swarm of adversarial agents. We employ a closed formation (`StringNet') of defending agents around the adversarial agents to restrict their motion and guide them to a safe area while navigating in an obstacle-populated environment. Control laws for forming the StringNet and guiding it to a safe area are developed, and the stability of the closed-loop system is analyzed formally. The adversarial swarm is assumed to move as a flock in the presence of rectangular obstacles. Simulation results are provided to demonstrate the efficacy of the approach.   
		
		
	\end{abstract}
	
	\section{Introduction}
	Swarm technology has seen a rapid growth recently. Safety-critical infrastructure such as government facilities, airports, military bases are at increased risk of being attacked by swarms of adversarial agents (e.g., aerial robots). This creates a need for defending safety-critical infrastructure from attacks of adversarial swarms, particularly in crowded urban areas. 
	
	Counteracting an adversarial swarm by means of physical interception \cite{chen2017multiplayer, coon2017control} in an urban environment may not be desired due to human presence. Under the assumption of risk-averse and self-interested adversarial agents (attackers) that tend to move away from the defending agents (defenders) and from other dynamic objects, herding can be used as an indirect way of guiding the attackers to a safe area.  
	
	In this paper, we consider a problem of defending a safety-critical area (protected area) from an adversarial swarm. We address this as a problem of herding a swarm of attackers to a safe area, while avoiding the static rectangular obstacles of the urban environment.
	
	The herding approach to herd a flock of birds away from an airport in \cite{gade2015herding} uses an $n$-wavefront algorithm, where the motion of the birds on the boundary of the flock is influenced based on the locations of the airport and the safe area. Stability and performance guarantees under directed star communication graph are provided in \cite{gade2016robotic}, and experimental results in \cite{paranjape2018robotic}. In \cite{pierson2018controlling} a circular arc formation of herders is used to influence the nonlinear dynamics of the herd based on a potential-field approach. The authors design a point-offset controller to guide the herd close to a specified location. In \cite{haque2011biologically}, biologically-inspired strategies are developed for confining a group of mobile robots. The authors develop strategies based on the "wall" and "encirclement" methods that dolphins use to capture a school of fish. Regions from which this confinement is possible are also derived; however, the results are limited to constant velocity motions. 
	A similar approach of herding by caging is adopted in \cite{varava2017herding}, where a cage of high potential is formed around the sheep (attackers). An RRT approach is used to find a motion plan for the agents while maintaining the cage. However, the formation is assumed to have been already formed around the sheep. Furthermore, the caging of the sheep is only ensured with constant velocity motion under additional conservative assumptions on the distances between the agents. In general, most of these works lack a proper modeling of the adversarial agents' intent to reach or attack a certain protected area.
	
	In \cite{licitra2017single,licitra2018single} the authors discuss herding using a switched systems approach; the herder (defender) chases targets (attackers) sequentially by switching among them so that certain dwell-time conditions are satisfied to guarantee stability of the resulting trajectories. However, the assumption that only one of the targets is influenced by the herder at any time is conservative for the problem of defending against a swarm of attackers. 
	The authors in \cite{deptula2018single} use approximate dynamic programming to obtain near-optimal control policies for the herder to chase a target agent to a goal location. 
	
	
	The aforementioned approaches assume some form of potential field to model the repulsion of the attackers from the defenders, and develop herding strategies for the defenders based on this potential field. Hence, such approaches may fail to create a proper potential barrier around the attackers if the potential field of the attackers is unknown to defenders, or is modeled inaccurately. 
	In addition, most of the earlier work does not consider obstacles in the environment. In our prior work \cite{chipade2019herding}, we developed a strategy for herding a single attacker to a safe area in the presence of rectangular obstacles.
	
	In this paper, we propose what we call `StringNet Herding', in which a closed formation of physical strings called `StringNet' is formed by the defenders around the swarm of attackers. It is assumed that the string between two defenders serves as a barrier through which the attackers cannot escape. The StringNet is controlled collectively to herd the swarm of attackers. The proposed approach only assumes that the attackers avoid collisions with defenders and barriers, while the control actions of the attackers are not known a priori. To demonstrate the proposed approach, we use flocking behavior for the attackers, which however is not known to the defenders.	
	
	We build on earlier work \cite{reynolds1987flocks, murray2003flocking,olfati2004flocking,dai2014flocking, tanner2007flocking} to develop a flocking controller for the attackers in the presence of rectangular obstacles. The controller uses the $\beta$-agent strategy \cite{murray2003flocking, olfati2004flocking}, in which a virtual agent called $\beta$-agent is assumed to move on the boundary of the obstacle, 
	and the control action is designed to maintain a certain distance from this $\beta$-agent using a potential function approach. We generate $\beta$-agents along a superelliptic curve that is at least $\calC^1$ around the rectangular obstacles. Also, in contrast to earlier work \cite{pierson2018controlling, varava2017herding} that treats robots as point masses, we assume agents with known circular footprints. Furthermore, no constant velocity assumption is made about the attackers as is done in  \cite{haque2011biologically,varava2017herding}. 

	In summary, the novelties and the contributions are: (i) A `StringNet' formation to restrict the motion of the attackers inside the StringNet and to herd them towards a safe area. We develop provably-correct control laws for the defenders to form the StringNet in finite time, and to herd the entrapped attackers to safe area. (ii) The definition of $\beta$-agents along superelliptic contours around rectangular obstacles with $\calC^1$ velocity profile for obstacle avoidance in flocking.
	

	The rest of the paper is structured as follows: Section \ref{sec:math_model} describes the mathematical modeling and problem statement. The flocking and herding algorithms are discussed in Section \ref{sec:flocking} and \ref{sec:herding}, while simulations are provided in Section \ref{sec:simulations}. The conclusions and the ongoing work are discussed in Section \ref{sec:conclusions}.
	
	\section{Modeling and Problem Statement}\label{sec:math_model}
	\textit{Notation}: Vectors and matrices are denoted by small and capital bold letters, respectively (e.g., $\textbf{r}$, $\textbf{P}$). Script letters denote sets (e.g., $\calP$). $\norm{.}$ denotes Euclidean norm of its argument. $\abs{.}$ denotes absolute value of a scalar argument and cardinality if the argument is a set. The function $\textbf{sig}^{\alpha}$ is defined as: $\textbf{sig}^{\alpha}(\mathbf{x})=\mathbf{x}\norm{\mathbf{x}}^{\alpha-1}$. $\bR_{\ge 0}=\{x\in \bR| x \ge 0\}$. 	
	$R_{\imath}^{\jmath}=\norm{\mathbf{r}_{\imath}-\mathbf{r}_{\jmath}}$ and $E_{ok}^{\imath}$ are the Euclidean distance between object $\jmath$ and $\imath$, and the Super-elliptic distance between $\imath$ and $\calO_k$, respectively. 
	A blending function~\cite{panagou2014motion}, $\sigma_{\imath}^{\jmath}(\delta)$, characterized by a doublet $(\underline{\delta}_{\imath}^{\jmath}, \bar{\delta}_{\imath}^{\jmath})$ with $\underline{\delta}_{\imath}^{\jmath} < \bar{\delta}_{\imath}^{\jmath}$ is defined as: 
	\setlength{\abovedisplayskip}{3pt}
	\setlength{\belowdisplayskip}{3pt}
	\be \label{eq:blendFun}
	\sigma_{\imath}^{\jmath}(\delta) =  
	\begin{cases}
		1, & \delta \le \underline{\delta}_{\imath}^{\jmath};	\\[3pt]
		A_{\imath}^{\jmath}\delta^3+B_{\imath}^{\jmath} \delta^2+ C_{\imath}^{\jmath} \delta+D_{\imath}^{\jmath}, &   \underline{\delta}_{\imath}^{\jmath} \le \delta \le \bar{\delta}_{\imath}^{\jmath};	\\
		0, &	 \delta \ge \bar{\delta}_{\imath}^{\jmath};
	\end{cases}
	\ee
	where $\delta$ is the distance between the objects $\imath$ and $\jmath$.
	The coefficients $A_{\imath}^{\jmath},B_{\imath}^{\jmath},C_{\imath}^{\jmath},D_{\imath}^{\jmath}$ are chosen as: $A_{\imath}^{\jmath}=\frac{2}{(\bar{\delta}_{\imath}^{\jmath}-\underline{\delta}_{\imath}^{\jmath})^3}$, $B_{\imath}^{\jmath}=\frac{-3(\bar{\delta}_{\imath}^{\jmath}+\underline{\delta}_{\imath}^{\jmath})}{(\bar{\delta}_{\imath}^{\jmath}-\underline{\delta}_{\imath}^{\jmath})^3}$,  $C_{\imath}^{\jmath}=\frac{6\bar{\delta}_{\imath}^{\jmath}\underline{\delta}_{\imath}^{\jmath}}{(\bar{\delta}_{\imath}^{\jmath}-\underline{\delta}_{\imath}^{\jmath})^3}$, $D_{\imath}^{\jmath}=\frac{(\bar{\delta}_{\imath}^{\jmath})^2(\bar{\delta}_{\imath}^{\jmath}-3\underline{\delta}_{\imath}^{\jmath})}{(\bar{\delta}_{\imath}^{\jmath}-\underline{\delta}_{\imath}^{\jmath})^3}$, so that \eqref{eq:blendFun} is a $\calC^1$ function. The argument $\delta$ is either the Euclidean distance or the Super-elliptic distance, depending on the objects under consideration, and will be omitted when clear from the context.	

	We consider $N_a$ attackers $\calA_i$, $i \in I_a= \{1,2,...,N_a\}$ and $N_d$ defenders $\calD_j$, $j \in I_d= \{1,2,...,N_d\}$, operating in a 2D environment $\calW \subseteq \mathbb{R}^2$ with $N_o$ rectangular obstacles, a protected area $\calP \subset \calW$ defined as $\calP=\{\textbf{r} \in \bR^2 \;| \; \norm{\textbf{r}-\textbf{r}_p}\le \rho_p\}$, and a safe area $\calS \subset \calW$, defined as $\calS=\{\textbf{r} \in \bR^2 \; | \; \norm{\textbf{r}-\textbf{r}_{s}}\le \rho_{s}\}$, where $(\textbf r_p, \rho_p)$ and $(\textbf r_{s}, \rho_{s})$ are the centers and radii of the corresponding areas, respectively. The agents $\calA_i$ and $\calD_j$ are modeled as discs of radii $\rho_a$ and $\rho_d\le \rho_a$, respectively and have Double Integrator (DI) dynamics with a linear drag term: 
	\be \label{eq:attackDyn1}
	\baa
	\dot{\textbf{r}}_{ai}
	=\textbf{v}_{ai}, \quad \quad 
	\dot{\textbf{v}}_{ai}
	=\textbf{u}_{ai}-C_{d}\textbf{v}_{ai};
	\eaa
	\ee	
	\be\label{eq:defendDyn1}	
	\baa
	\dot{\textbf{r}}_{dj}
	=\textbf{v}_{dj}, \quad \quad 
	\dot{\textbf{v}}_{dj}
	=\textbf{u}_{dj}-C_{d}\textbf{v}_{dj};
	\eaa
	\ee	
	where $C_d$ is a drag coefficient, for $\imath=ai$ and $\imath=dj$
	$\textbf{r}_{\imath}=[x_{\imath}\; y_{\imath}]^T$, $\textbf{v}_{\imath}=[v_{x_{\imath}}\; v_{y_{\imath}}]^T$ are position and velocity of $\calA_i$ and $\calD_j$, respectively, and $\textbf{u}_{\imath}=[u_{x_{\imath}}\; u_{y_{\imath}}]^T$ is acceleration input (control input) of $\calA_i$ and $\calD_j$, respectively.	
	We assume that the control action of $\calA_i$ satisfies $\norm{\mathbf{u}_{ai}}<u_{ma}$.
	This model poses a realistic speed bound on each attacker with limited acceleration control, i.e., $v_{ai}=\norm{\mathbf{v}_{ai}}<v_{ma}=\frac{u_{ma}}{C_d}$. 
	We assume that every defender $\calD_j$ senses the position $\textbf r_{ai}$ and velocity $\textbf{v}_{ai}$ of the attacker $\calA_i$ when $\calA_i$ is inside a circular sensing-zone $\calZ_d^s=\{\mathbf r \in \mathbb{R}^2 |\; \norm{\textbf{r}-\textbf{r}_p} \le \rho_d^s\}$ around $\calP$. Each attacker $\calA_i$ has a similar local sensing zone $\calZ_{ai}^s=\{\textbf{r} \in \bR^2 \;| \; \norm{\textbf{r}-\textbf{r}_{ai}}\le \rho_{ai}^s \}$.
	
	We consider static obstacles $\calO_k$ of rectangular shape, with their edges along the $x$-axis ($\hat{\textbf{i}}$) and $y$-axis ($\hat{\textbf{j}}$) of  a coordinate frame $\calF_{gi}$,   
	defined as:
	\be
	\calO_k= \{\mathbf r \in \mathbb{R}^2 | \abs{x-x_{ok}} \le \frac{w_{ok}}{2}, \abs{y-y_{ok}} \le \frac{h_{ok}}{2}\} ,
	\ee
	where $\mathbf r_{ok}=[x_{ok} \; y_{ok}]^T$ is the center, $w_{ok}$ and $h_{ok}$ are the lengths along $\hat{\textbf{i}}$ and $\hat{\textbf{j}}$ of $\calO_k$ for all $ k \in I_o = \{1,2,...,N_o\}$.

	The attackers aim to reach the protected area $\calP$ as a flock, and the defenders aim to herd the flock to the safe area $\calS$ before it reaches $\calP$. Formally, we consider the following two problems.
	\begin{problem}[Flocking]
		Design control actions $\mathbf{u}_{ai}$, $\forall i \in I_a$ such that $\calA$'s reach $\calP$ as a flock formation while avoiding the static rectangular obstacles.
	\end{problem}
	\begin{problem}[Herding]
		Find control actions $\mathbf{u}_{dj}$, $\forall j \in I_d$ to accomplish: 1) StringNet formation around the swarm of attackers in finite time. 2) Once the StringNet is formed, move the StringNet to the safe area $\calS$ while avoiding the obstacles $\calO_k$. 
	\end{problem}

	\section{Flocking}\label{sec:flocking}
	The neighboring graph \cite{tanner2007flocking} for the attackers is denoted as $\calG_a=(\calV_a, \calE_a)$, where $\calV_a =\{\calA_1,\calA_2,...,\calA_{N_a}\}$ is the set of vertices and $\calE_a$ is the set of edges. Each attacker $\calA_i$ communicates with its neighbors $\calN_{ai}^{a} = \{i'\in \calV_a | (\calA_i,\calA_{i'}) \in \calE_a \}$.
	We define a potential function $V_{\imath}^{\jmath}:\bR_{\ge 0}\rightarrow \bR_{\ge 0}$ as:
	\be \noindent \label{eq:potential_function1}
	\arraycolsep=-3pt
	V_{\imath}^{\jmath}(R_{\imath}^{\jmath})=
	\ln \left(\frac{\tilde{R}_\imath^\jmath}{R_\imath^\jmath -\hat{R}_{\imath}^{\jmath}}+\frac{R_\imath^\jmath -\hat{R}_{\imath}^{\jmath}}{\tilde{R}_\imath^\jmath}\right),
	\ee
	where $\tilde{R}_\imath^\jmath>\hat{R}_{\imath}^{\jmath}$ is the desired distance between agent $\imath$ and agent $\jmath$, and $\hat{R}_{\imath}^{\jmath}$ is the minimum safety distance between agent $\imath$ and agent $\jmath$. We have that as $R_\imath^\jmath$ approaches $\hat{R}_{\imath}^{\jmath}$, the potential $V_{\imath}^{\jmath}$ tends to $\infty$.
	A control action corresponding to $V_\imath^\jmath$ is defined as:
	\be \label{eq:potential_function_control}
	\mathbf{u}_{p}(\mathbf{x}_\imath^\jmath)= -\zeta_{\imath}^{\jmath}   (\mathbf{v}_{\imath}-\mathbf{v}_{\jmath}) -\mu_{\imath}^{\jmath} \cdot \nabla_{\mathbf{r}_{\imath}} V_{\imath}^{\jmath}
	\ee
	where $\mathbf{x}_\imath^\jmath=[\mathbf{r}_\imath^T,  \mathbf{v}_\imath^T, \mathbf{r}_{\jmath}^T, \mathbf{v}_{\jmath}^T]^T$, $\zeta_{\imath}^{\jmath}$ and $\mu_\imath^\jmath$ are control gains. 

	The swarm of attackers aims to reach the protected area $\calP$ while avoiding the static obstacles $\calO_k$ and maintaining a flock described by potential functions $V_{ai}^{ai'}$ for all $i,i' \in I_a$ over the graph $\calG_a$. The control action for the flock of the attackers is defined as \cite{tanner2007flocking,deghat2016combined}:
	\setlength{\abovedisplayskip}{3pt}
	\setlength{\belowdisplayskip}{3pt}
	\be \label{eq:att_flock_control1}
	\arraycolsep=1.4pt
	\baa	
	\mathbf{u}_{ai}^f =& k_{a}^r(\mathbf{r}_{p}-\mathbf{r}_{ai})+
	\displaystyle \sum_{i'\in \calN_{ai}^a} \mathbf{u}_{p}(\mathbf{x}_{ai}^{ai'}) +\displaystyle \sum_{k\in\calN_{ai}^o} \sigma_{ai}^{ok} \cdot \mathbf{u}_{p}(\mathbf{x}_{ai}^{\beta ik})
	\eaa
	\ee
	where $\mathbf{x}_{ai}^{\beta ik}=[\mathbf{r}_{ai}^T, \mathbf{v}_{ai}^T,\mathbf{r}_{\beta ik}^T, \mathbf{v}_{\beta ik}^T]^T$, where $\mathbf{r}_{\beta ik}$ and $\mathbf{v}_{\beta ik}$ is the position and velocity of the $\beta$-agent on the boundary of $\calO_k$ corresponding to $\calA_i$ that is used for avoiding $\calO_k$. The blending function $\sigma_{ai}^{ok}$ allows smooth transition to the obstacle avoidance part of the controller, and is characterized by the doublet $(\underline{\xi}_{a}^{o},\bar{\xi}_{a}^{o})$. $\calN_{ai}^o$ is a set of neighboring obstacles defined as: $\calN_{ai}^o=\{k \in I_o| \sigma_{ai}^{ok}> 0\}$. The center $\mathbf{r}_p$ of the protected area $\calP$ acts as a $\gamma$-agent \cite{olfati2004flocking} providing navigational feedback. 
	
	
	\subsection{$\beta$-agents around Rectangular Obstacles} 
	
	\begin{wrapfigure}{l}{0.25\textwidth}
		\setlength{\belowcaptionskip}{-14pt plus 3pt minus 2pt}	
		\includegraphics[width=1\linewidth,trim={1.9cm 1cm 2.3cm 0cm},clip]{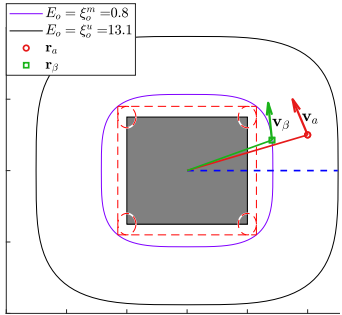}
		\caption{$\beta$-agent around rectangles for obstacle avoidance}
		\label{fig:pos_beta_agent}
	\end{wrapfigure}
	The position $\mathbf{r}_{\beta ik}$ $=[x_{\beta ik},\; y_{\beta ik}]^T$ of the agent $\beta_{ik}$ is defined as the projection of $\mathbf{r}_{ai}$ on the superelliptic contour $\calSE_{ok}$ of level $\xi_{ok}^{m}$ \cite{volpe1990manipulator}, defined as: $ 
	\calSE_{ok}= \left \{\mathbf{r} \in \bR^2 | E_{ok} =\xi_{ok}^{m} \right \}$ and the velocity $\mathbf{v}_{\beta ik}$ as the projection of $\mathbf{v}_{ai}$ along the tangent to the superelliptic contour $\calSE_{ok}$ at $\mathbf{r}_{\beta i}$ along the direction of motion of $\calA_i$.

	\noindent The superelliptic distance $E_{ok}$ is defined as:
	\be \label{eq:superEllipticDist}
	E_{ok}= \abs{\frac{x-x_{ok}}{a_{ok}}}^{2n_{ok}} + \abs{\frac{y-y_{ok}}{b_{ok}}}^{2n_{ok}} -1.
	\ee
	The projection $\mathbf{r}_{\beta ik}$ of $\mathbf{r}_{ai}$ on the $\calSE_{ok}$ is the closest point on $\calSE_{ok}$ such that the unit tangent $\hat{\mathbf{t}}_{ok}(\mathbf{r}_{\beta ik})$ to $\calSE_{ok}$ at $\mathbf{r}_{\beta ik}$ is normal to $\mathbf{r}_{ai}-\mathbf{r}_{\beta ik}$, and is found by solving:
	\begin{subequations} \label{eq:pos_beta_agent}
		\be
		\abs{\frac{x_{\beta ik}-x_{ok}}{a_{ok}}}^{2n_{ok}} + \abs{\frac{y_{\beta ik}-y_{ok}}{b_{ok}}}^{2n_{ok}} -1 =\xi_{ok}^{m},
		\ee 
		\be -\frac{b^{2n_{ok}}sig^{2n_{ok}-2}(x_{\beta ik}-x_{ok})}{a^{2n_{ok}}sig^{2n_{ok}-2}(y_{\beta ik}-y_{ok})}\cdot\frac{y_{\beta ik}-y_{ai}}{x_{\beta ik}-x_{ai}} =-1. 
		\ee
	\end{subequations}
	where $sig^m(x)=x\abs{x}^m$. Fig. \ref{fig:pos_beta_agent} shows the projection $\mathbf{r}_{\beta ik}$ (green square) of $\mathbf{r}_{ai}$ (red circle). The velocity $\mathbf{v}_{\beta ik}$ can be then obtained as: $
	\mathbf{v}_{\beta ik}=\left ( \mathbf{v}_{ai}\cdot\hat{\mathbf{t}}_{ok}(\mathbf{r}_{\beta ik})\right) \hat{\mathbf{t}}_{ok}(\mathbf{r}_{\beta ik})
	$.	

	\subsection{Avoiding Dynamic Obstacles during Flocking}
	\subsubsection{Avoiding the Defenders}
	In addition to avoiding static obstacles, the attackers apply the following control action to avoid the defenders:
	\be \label{eq:att_flock_control2}
	\baa
	\mathbf{u}_{ai}^d =\displaystyle \sum_{j\in \calN_{ai}^{d}} \sigma_{ai}^{dj} \cdot \mathbf{u}_{p}(\mathbf{x}_{ai}^{dj}),
	\eaa
	\ee
	where $\sigma_{ai}^{dj}$ is characterized by the doublet $(\underline{R}_{a}^{d},\bar{R}_{a}^{d})$. $\calN_{ai}^d$ is a set of defenders in the sensing zone of $\calA_i$ defined as: $\calN_{ai}^d=\{j \in I_d| R_{ai}^{dj}< \rho_{ai}^s\}$, and $\tilde{R}_{ai}^{dj}>\bar{R}_{a}^{d}$.
	
	\subsubsection{Avoiding the Strings}
	The strings (string barriers) are line segments between defenders. The attackers can sense these strings in their sensing zone and react to them using the control action:
	\be \label{eq:att_flock_control3}
	\baa
	\mathbf{u}_{ai}^{b} =\displaystyle \sum_{s\in \calN_{ai}^{b}} \sigma_{ai}^{bs} \cdot \mathbf{u}_{p}(\mathbf{x}_{ai}^{bs}),
	\eaa
	\ee
	where $\mathbf{u}_{p}$ is given by \eqref{eq:potential_function_control} and $V_{ai}^{bs}$ is a potential function for $\calA_i$ corresponding to its projection ($\mathbf{r}_{bs}$, $\mathbf{v}_{bs}$) on the string barrier $\calB_{s}$ (Fig.~\ref{fig:stringNetDesPos}). $\sigma_{ai}^{bs}$ and $\calN_{ai}^{b}$ are defined similar to $\sigma_{ai}^{dj}$ and $\calN_{ai}^{d}$.
	The combined bounded control action for the attackers' flock is given as:
	\be \label{eq:att_flock_control}
	\baa
	\mathbf{u}_{ai} =& \bm{\sigma}_a \left (\mathbf{u}_{ai}^f + \mathbf{u}_{ai}^d + \mathbf{u}_{ai}^b  +C_d\mathbf{v}_{ai} \right),
	\eaa
	\ee
	where saturation function $\bm{\sigma}_a(\mathbf{u})=\min(u_{ma},\norm{\mathbf{u}})\frac{\mathbf{u}}{\norm{\mathbf{u}}}$.
	\begin{remark}
		The convergence analysis for flocking of the attackers under the control \eqref{eq:att_flock_control1} is provided in \cite{olfati2004flocking,murray2003flocking,tanner2007flocking} when the first term is absent i.e. no navigational control command. Similar analysis can be done for the flock's convergence to $\mathbf{r}_p$. Since flocking is not the focus of this paper we omit the analysis in the interest of space.
	\end{remark}

	\section{Herding}\label{sec:herding}
	To herd the flock of attackers to $\calS$, we propose `StringNet Herding'. StringNet is a closed net of strings formed by the defenders as shown in Fig. \ref{fig:stringNetDesPos}. The strings can be actual physical strings (ropes) or some mechanism that does not allow the attackers to pass through them. It is assumed that even after being connected by the strings, the motion of defenders is not restricted. The underlying graph structure for the `StringNet' is defined as:
	\begin{definition}[StringNet] The StringNet $\calG^{s}= (\calV^s,$ $\calE^s)$ is a cycle graph consisting of: 1) the defenders as the vertices, $\calV^s=\{\calD_1,\calD_2,...,\calD_{N_d}\}$, 2) a set of edges, $\calE^s=\{(\calD_j,\calD_{j'}) \in \calV^s \times \calV^s | \calD_j \overset{s} \longleftrightarrow \calD_{j'} \}$. The operator $\overset{s} \longleftrightarrow$ denotes a physical string barrier between the defenders.
	\end{definition}
	
	The StringNet herding consists of three phases: 1) Gathering, 2) StringNet formation and 3) Herding the StringNet to $\calS$. These phases are discussed as follows.
	
	\subsection{Gathering}
	Once the adversarial attackers are sensed in the sensing zone $\calZ_{d}^s$, the defenders are tasked to herd them. The defenders first converge to an open semicircular formation in the expected path of the attackers (shortest path for the attackers) and establish strings such that $\calA_i$ is connected to $\calA_{i+1}$ by a string for all $i=\{1,2,...,N_d-1\}$ (Fig.~\ref{fig:stringNetDesPos}). The desired position $\bm{\xi}_{dj}^g$ of $\calD_j$ on the stationary semicircular formation $\mathscr{F}_d^g$ (Fig.~\ref{fig:stringNetDesPos}) is designed as:
	\be 
	\arraycolsep=1.4pt
	\baa
	\bm{\xi}_{dj}^g=\mathbf{r}_{df}^g + \rho_{df}^s \hat{\mathbf{o}} (\theta_{dj}) \text{, where }
	\theta_{dj} = \theta_{df}^{g*}+\frac{\pi(j-1)}{N_d-1},
	\eaa
	\ee
	where $\hat{\mathbf{o}}(\theta)=\bbmat \cos(\theta) \\ \sin(\theta) \ebmat$ is the unit vector making an angle $\theta$ with $x$-axis, $\mathbf{r}_{df}^g = \rho_{df}^g \hat{\mathbf{o}}(\theta_{ac}^*) $ is a location such that $\rho_{df}^g> \rho_p + d_{ac}^{max}$, where $d_{ac}^{max}$ is the maximum distance attacker's center of mass (ACoM, $\mathbf{r}_{ac} =\frac{1}{N_a}\sum_{i}^{N_a}\mathbf{r}_{ai}$) can travel towards $\calP$ during the StringNet formation phase, discussed next, and $\theta_{df}^{g*}=\theta_{ac}^{*}-\frac{\pi}{2}$, where $\theta_{ac}^*$ is the expected direction of motion of the ACoM  on the shortest path from the initial position of ACoM to $\calP$. We have $\dot{\bm{\xi}}_{dj}^g=	\bm{\eta}_{dj}^g=\mathbf{0}$ and $\dot{\bm{\eta}}_{dj}^g=\mathbf{0}$. We assume the following.
	\begin{assumption}\label{assum:gathering_phase_params}
		(a) The desired position of $\calD_j$, $\bm{\xi}_{dj}^g$, is such that $E_{ok}^{dj,des}>\bar{\xi}_{d}^o, \forall j \in I_d, \forall k \in I_o$, where $E_{ok}^{dj,des}$ is super-elliptic distance between $\bm{\xi}_{dj}^g$ and the obstacle $\calO_k$. (b)
		$\rho_{df}^s\sqrt{2-2\cos \left(\frac{\pi}{N_d-1}\right)} > \bar{R}_{d}^{d}$, $\rho_{df}^s> \rho_{ac}+ \bar{R}_{d}^{\delta c}$ where $\bar{R}_{d}^{d}$ and $\bar{R}_{d}^{\delta c}$ are the parameters of the blending functions $\sigma_{dj}^{dj'}$ and $\sigma_{dj}^{\delta cj}$ respectively.
	\end{assumption}
	
	To converge to $\bm{\xi}_{dj}^g$, a finite-time stabilizing controller is defined as:
	\be \label{eq:def_control1} 
	\mathbf{u}_{dj}=\mathbf{u}_{dj}^0 +\mathbf{u}_{dj}^{col} + \dot{\bm{\eta}}_{dj}^g\vspace{-5pt},
	\ee
	where
	\bewn
	\mathbf{u}_{dj}^0= C_d \mathbf{v}_{dj} - k_2\textbf{sig}^{\alpha_2}(\mathbf{v}_{dj}-\bm{\eta}_{dj}^g)- k_1\textbf{sig}^{\alpha_1}(\mathbf{r}_{dj}-\bm{\xi}_{dj}^g)
	\eewn \vspace{-10pt}	 
	\be \label{eq:def_collision_avoid_control1}	 
	\mathbf{u}_{dj}^{col} = \displaystyle \sum_{j'\in \calN_{dj}^{d}} \sigma_{dj}^{dj'} \cdot \mathbf{u}_{p}(\mathbf{x}_{dj}^{dj'})+\displaystyle \sum_{k \in \calN_{dj}^o}^{} \sigma_{dj}^{\delta jk} \cdot \mathbf{u}_{p}(\mathbf{x}_{dj}^{\delta jk}) \vspace{-10pt}
	\ee
	where $k_1, k_2 >0$. $\mathbf{x}_{dj}^{\delta jk}=[\mathbf{r}_{dj}^T,\mathbf{v}_{dj}^T,\mathbf{r}_{\delta jk}^T,\mathbf{v}_{\delta jk}^T]$, where $\mathbf{r}_{\delta jk}$ and $\mathbf{v}_{\delta jk}$ are the position and the velocity of a virtual $\delta$-agent, similar to $\beta$-agent, corresponding to $\calD_j$ around the obstacle $\calO_k$. $V_{dj}^{dj'}$, $V_{dj}^{\delta jk}$ are potential functions to avoid collision, respectively, with $\calD_{j'}$ and $\delta$-agent on the boundary of $\calO_k$. We have $\tilde{R}_{dj}^{dj'}>\bar{R}_{d}^{d}$ and $\tilde{R}_{dj}^{\delta jk}>\bar{R}_{d}^{d}$ to ensure collision avoidance for $\calD_j$.

	\subsection{StringNet Formation}
	The attackers are assumed to stay within a connectivity region of radius $\rho_{ac}$ ($<\rho_{sn}^{max} $) around ACoM . Once the semicircular formation is in place, the defenders wait until attackers come close, i.e., $\norm{\mathbf{r}_{df}^g-\mathbf{r}_{ac}}<\epsilon$, where $\epsilon$ is a small number. 
	To trap the attackers inside StringNet, a desired regular-polygon formation is designed around the connectivity region of the attackers as shown in Fig.~\ref{fig:stringNetDesPos}. The defenders start tracking their desired positions around the attackers and once $\calD_1$ and $\calD_{N_a}$ reach within $b_d$ distance from their respective desired positions they get connected via a string. The desired position $\bm{\xi}_{dj}^s$ of $\calD_j$ on the StringNet $\calG^s$ (Fig.~\ref{fig:stringNetDesPos}) is chosen on the circle with radius $\rho_{sn}$ centered at $\mathbf{r}_{ac}$ as:
	\be 
	\arraycolsep=1.4pt
	\baa
	\bm{\xi}_{dj}^s=\mathbf{r}_{ac} + \rho_{sn} \hat{\mathbf{o}}(\theta_{dj}) \text{, where }
	\theta_{dj} = \theta_{df}^{s*}+\frac{\pi(2j-1)}{N_d},
	\eaa
	\ee
	for all $j \in I_d$, where $\theta_{df}^{s*}=\theta_{df}^{g*}$. 
	The radius $\rho_{sn}$ should satisfy, $\rho_{ac} +b_d < \rho_{sn}\le \rho_{sn}^{max}-b_d$, where $\rho_{sn}^{max}$ is the maximum footprint of a formation that can pass through the obstacle-free space in the environment. The parameter $b_d$ is the maximum position tracking error when the defenders converge to the StringNet formation as obtained in Theorem~\ref{thm:stringNet_formation}. We have $\dot{\bm{\xi}}_{dj}^s=\bm{\eta}_{dj}^s=\dot{\mathbf{r}}_{ac} =\mathbf{v}_{ac}$.
	\begin{figure}[h]
		\centering
		\includegraphics[width=.9\linewidth,trim={3.2cm 3.7cm 2.5cm 3cm},clip]{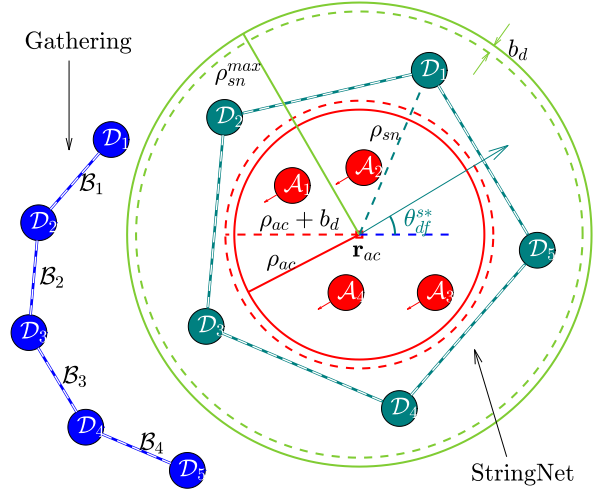}
		\caption{Desired Positions of the Defenders}
		\label{fig:stringNetDesPos}
	\end{figure}	
	The control action for $\calD_j$ during this phase is:
	\be\label{eq:def_control2}
	\arraycolsep=1.4pt
	\baa
	\mathbf{u}_{dj}=& C_d\mathbf{v}_{dj}- k_2\cdot(\mathbf{v}_{dj}-\bm{\eta}_{dj}^s)- k_1\cdot(\mathbf{r}_{dj}-\bm{\xi}_{dj}^s)\\
	&+\sigma_{dj}^{\delta cj} \cdot \mathbf{u}_{p}(\mathbf{x}_{dj}^{\delta cj})+\mathbf{u}_{dj}^{col},
	\eaa
	\ee	
	where $\mathbf{r}_{\delta cj}$ and $\mathbf{v}_{\delta cj}$ are the position and the velocity of the $\delta$-agent corresponding to the $\calD_j$ on the boundary of the connectivity region of the attackers. The StringNet is achieved when $\norm{\mathbf{r}_{dj}-\mathbf{\bm{\xi}}_{dj}} \le b_d$ for all $j \in I_d$ during this phase.
	To ensure enough space for the movement of the attackers inside the StringNet, the minimum number of defenders require to herd the given number of attackers with connectivity region of radius $\rho_{ac}$ is: $
	N_d^{min}=\ceil[\Bigg]{\frac{\pi}{\cos^{-1} \left (\frac{\rho_{ac}+b_d}{\rho_{sn}^{max}-b_d}\right)}} $,
	where $\ceil*{\cdot}$ gives the smallest integer greater than or equal to its argument.
	\subsection{Herding: Moving the StringNet to safe area}
	Once the defenders form a StringNet around the attackers, they move while tracking a desired rigid closed formation $\mathscr{F}_d^h$ centered at a virtual agent $\mathbf{r}_{df}$. The virtual agent's dynamics are governed by the DI dynamics similar to \eqref{eq:defendDyn1} with acceleration, 
	\be \label{eq:def_des_control3}	
	\mathbf{u}_{df}=\bm{\sigma}_{d_h} \biggl ( - k_1(\mathbf{r}_{df}-\mathbf{r}_{s})  +\displaystyle \sum_{k \in  \calN_{df}^o}^{} \sigma_{df}^{\delta fk}  \mathbf{u}_{p}(\mathbf{x}_{df}^{\delta fk}) \biggr),
	\ee
	where $\delta {fk}$ refers to the $\delta$-agent on the obstacle $\calO_k$ corresponding to virtual agent at $\mathbf{r}_{df}$, and $\bm{\sigma}_{d_h}(\mathbf{u})=\min(u_{md}^h,\norm{\mathbf{u}})\frac{\mathbf{u}}{\norm{\mathbf{u}}}$. We choose $u_{md}^h < u_{ma}$ to ensure that the attackers are able to react to the motion of the defenders. The desired positions $\bm{\xi}_{dj}^h$ of the defenders on the desired closed formation $\mathscr{F}_d^h$ satisfy:
	\be \label{eq:desired_formation_herding}
	\arraycolsep=1.4pt
	\baa
	\dot{\bm{\xi}}_{dj}^h=\bm{\eta}_{dj}^h, \quad \quad \dot{\bm{\eta}}_{dj}^h=\mathbf{u}_{df}-C_d \mathbf{v}_{df};\\
	\bm{\xi}_{dj}^h=\mathbf{r}_{df} + \rho_{sn} \hat{\mathbf{o}}(\theta_{dj}) \text{, where }
	\theta_{dj} = \theta_{df}^{s*}+\frac{\pi(2j-1)}{N_d}.
	\eaa
	\ee
	The control~\eqref{eq:def_control1} is appropriately modified to track $(\bm{\xi}_{dj}^h,$ $\bm{\eta}_{dj}^h)$ by replacing $\bm{\xi}_{dj}^s,\bm{\eta}_{dj}^s,\dot{\bm{\eta}}_{dj}^s$ by $\bm{\xi}_{dj}^h,\bm{\eta}_{dj}^h,\dot{\bm{\eta}}_{dj}^h$, respectively.
	\subsection{Convergence Analysis}
	\begin{theorem} \label{thm:stringNet_formation}
		The StringNet $\calG^s$ centered at $\mathbf{r}_{ac}$ is formed around the attackers in finite time from almost all initial conditions under the control action in~\eqref{eq:def_control1} (gathering phase) and \eqref{eq:def_control2} (StringNet formation phase), while avoiding collisions.
	\end{theorem}
	\begin{proof}
		For almost all initial conditions\footnote{Except for those in the set $\mathcal{M}_0=\{\mathbf{r}_{dj}, \mathbf{v}_{dj} \in \bR^2$ $ \; \forall j \in I_d | \mathbf{v}_{dj} = \mathbf{0}, \mathbf{u}_{dj}  = \mathbf{0} \text{ as per \eqref{eq:def_control1}}, \eqref{eq:def_control2} \}$, and the initial conditions from which the defenders' trajectories approach $\calM_0$; the latter depends on the desired states. A formal characterization of this set is left open for future research.} such that $R_{dj}^{dj'} >\hat{R}_{dj}^{dj'}$, we have $\frac{\partial V_{dj}^{dj'}}{\partial R_{dj}^{dj'}} \rightarrow \infty $ as $R_{dj}^{dj'} \rightarrow \hat{R}_{dj}^{dj'}$ implying infinite acceleration applied on $\calD_j$ in the direction away from $\calD_j'$ which ensures $R_{dj}^{dj'} >\hat{R}_{dj}^{dj'}$ at all times and hence ensures no collision among the defenders. A similar argument can be used to show obstacle avoidance. 

		During the gathering phase, when the defenders are not in conflict with other defenders or obstacle (i.e., $\sigma_{dj}^{dj'}=\sigma_{dj}^{\delta jk}=0$, $\forall j, j' \in I_d; k \in I_o$), the dynamics read:		
		\be \label{eq:def_control4}
		\baa
		\dot{\mathbf{r}}_{dj}=\mathbf{v}_{dj}\\
		\dot{\mathbf{v}}_{dj}= - k_2 \mathbf{sig}^{\alpha_{2}}(\mathbf{v}_{dj}-\bm{\eta}_{dj}^g)- k_{1}\mathbf{sig}^{\alpha_{1}}(\mathbf{r}_{dj}-\bm{\xi}_{dj}^g)
		\eaa
		\ee
		The origin $\mathbf{r}_{dj}-\bm{\xi}_{dj}^g= \mathbf{v}_{dj}=\mathbf{0}$ of \eqref{eq:def_control4} is finite-time stable \cite{bhat2005geometric} if $\alpha_1=\frac{\alpha_2}{2-\alpha_2}$. Let the convergence time be $T_{d}^{g}$.
		
		Similarly during the StringNet formation phase, when $\calD_j$ is not in conflict with any other defenders or obstacle, the error dynamics read:
		\be \label{eq:def_err_dyn1}
		\arraycolsep=0pt
		\dot{\mathbf{e}}_{dj}=
		\bbmat
		\dot{\mathbf{e}}_{dj}^r \\
		\dot{\mathbf{e}}_{dj}^v
		\ebmat= 
		\begingroup
		\setlength\arraycolsep{2pt}
		\bbmat 0 & 1\\ -k_1 & -k_2\ebmat
		\endgroup
		\bbmat 
		\mathbf{e}_{dj}^r \\\mathbf{e}_{dj}^v 
		\ebmat
		+ \bbmat
		0 \\
		\dot{\bm{\eta}}_{dj}^s
		\ebmat
		=\mathbf{A}\mathbf{e}_{dj} + \mathbf{g}_{dj}
		\ee
		where $\mathbf{e}_{dj}^r=\mathbf{r}_{dj}-\bm{\xi}_{dj}^s$, $\mathbf{e}_{dj}^v=\mathbf{v}_{dj}-\bm{\eta}_{dj}^s$, and $\norm{\dot{\bm{\eta}}_{dj}^s}=\norm{\dot{\mathbf{v}}_{ac}} \le u_{ma}$ which implies the disturbance term $\mathbf{g}_{dj}$ is bounded: $\norm{\mathbf{g}_{dj}}\le u_{ma}$. The nominal system in~\eqref{eq:def_err_dyn1}, $\dot{\mathbf{e}}_{dj}=\mathbf{A}\mathbf{e}_{dj}$, is exponentially stable for $k_1$, $k_2 >0$. From Theorem 4.6 in \cite{khalil2015nonlinear},
		there exists a positive definite matrix $\mathbf{P}$ that satisfies the Lyapunov equation $\mathbf{A}^T\mathbf{P}+\mathbf{P}\mathbf{A}=-\mathbf{Q}$, for any given positive definite matrix $\mathbf{Q}$. The Lyapunov function  $V_{dj}=\mathbf{e}_{dj}^T \mathbf{P} \mathbf{e}_{dj}$ satisfies the conditions as required in Lemma 9.2 in \cite{khalil2015nonlinear} with constants $c_1, c_2, c_3, c_4$ given in terms of the eigenvalues of $\mathbf{P}$ and $\mathbf{Q}$ as:
		$c_1=\lambda_{min}(\mathbf{P})$, $c_2=\lambda_{max}(\mathbf{P})$, $c_3=\lambda_{min}(\mathbf{Q})$ and $c_4=2 \lambda_{max}(\mathbf{P})$. 
		From Lemma 9.2 in \cite{khalil2015nonlinear}, if $\norm{\mathbf{g}_{dj}}\le u_{ma} < \frac{c_3}{c_4}\sqrt{\frac{c_1}{c_2}}c_0\bar{e}$ for all $t>0$, all $\mathbf{e}_{dj} \in D=\{\mathbf{e}_{dj} \in \bR^4| \norm{\mathbf{e}_{dj}} < \bar{e}\}$ with $c_0<1$, then for all $\norm{\mathbf{e}_{dj}(0)} <\sqrt{\frac{c_1}{c_2}}\bar{e}$, the solution $\mathbf{e}_{dj}(t)$ of the perturbed system in~\eqref{eq:def_err_dyn1} satisfies:
		\bi
		\item[1)] $\frac{\norm{\mathbf{e}_{dj}(t)}}{\norm{\mathbf{e}_{dj}(t_0)}} \le \sqrt{\frac{c_2}{c_1}} e^{\left(-\frac{(1-c_0)c_3}{2c_2}(t-t_0)\right)}$, $ \forall t_0 \le t <t_0+T_{dj}$,
		\item[2)] $\norm{\mathbf{e}_{dj}(t)} \le b_{dj} = \frac{c_4}{c_3}\sqrt{\frac{c_2}{c_1}}\frac{u_{ma}}{c_0}$, $\forall t \ge t_0+T_{dj}$,
		\ei
		for some finite time $T_{dj}$. That is, $\calD_j$ tracks the desired trajectory ($\bm{\xi}_{dj}^s$, $\bm{\eta}_{dj}^s$) within the error bound $b_{dj}$. Denote $b_d=\max_{j \in I_d} {b_{dj}}$. After the first two phases, all the defenders reach their desired locations within $b_{d}$ neighborhood in finite time $T \ge T_{d}^{g}+ \max_{j \in I_d} T_{dj}$ and hence the StringNet is formed in finite time.
	\end{proof}
	\begin{remark}
		All the attackers get entrapped inside the StringNet if the defenders form $\mathscr{F}_d^g$ before the attackers reach within $\rho_{df}^g+\rho_{sn}$ distance from $\calP$.
	\end{remark}
	\begin{theorem}
		Once the defenders form the StringNet $\calG^s$, they herd all the attackers trapped inside $\calG^s$ to the safe area $\calS \; (\rho_s>\rho_{sn}^{max})$ while avoiding the obstacles by tracking desired positions governed by \eqref{eq:def_des_control3} under the appropriately modified control action in \eqref{eq:def_control1}.
	\end{theorem}
	\begin{proof}
		Since the desired formation $\mathscr{F}_d^h$ moves as a rigid formation, we only consider the virtual agent at $\mathbf{r}_{df}$ with size $\rho_{sn}+\rho_d$ whose dynamics are:
		\be \label{eq:desired_formation_dyn}
		\dot{\bar{\mathbf{r}}}_{df} = \mathbf{v}_{df}, \quad \quad
		\dot{\mathbf{v}}_{df} = \bm{\sigma}_{d_h}(\mathbf{u}_{df}^h)-C_d\mathbf{v}_{df},
		\ee 
		where $\mathbf{u}_{df}^h= - k_1\bar{\mathbf{r}}_{df}  +\displaystyle \sum_{k \in  \calN_{df}^o}^{} \sigma_{df}^{\delta fk}  \mathbf{u}_{p}(\mathbf{x}_{df}^{\delta fk})$ and $\bar{\mathbf{r}}_{df}=\mathbf{r}_{df}-\mathbf{r}_{s}$. Using similar arguments as in Theorem \ref{thm:stringNet_formation}, we can ensure the safety of $\mathscr{F}_d^h$ if $\hat{R}_{df}^{\delta ok}>\rho_{sn}+\rho_{d}+\bar{\rho}$, where $\bar{\rho}=\frac{u_{md}^h(1-\log(2))}{C_d^2}$ is the maximum distance the formation can travel in the worst case motion of the formation toward the obstacle with the bounded acceleration. The formation will leave the locally active potential fields around the static obstacle in some finite time. In the absence of any obstacle's local potential field, we have $\mathbf{u}_{df}^h= - k_1\bar{\mathbf{r}}_{df}$. 
		We define a candidate Lyapunov function:
		\be
		\arraycolsep=0pt
		V=\begin{cases}
			\frac{k_1 \norm{\bar{\mathbf{r}}_{df}}^2}{2}+\frac{\norm{\mathbf{v}_{df}}^2}{2}, & \text{ if } \norm{\bar{\mathbf{r}}_{df}}<\frac{u_{md}^h}{k_1},\\
			u_{md}^h\norm{\bar{\mathbf{r}}_{df}}+\frac{\norm{\mathbf{v}_{df}}^2}{2} - \frac{(u_{md}^h)^2}{2k_1}, & \text{ otherwise}.
		\end{cases}
		\ee
		$V$ is 0 at $\bar{\mathbf{r}}_{df}=\mathbf{v}_{df}=\mathbf{0}$, is positive definite, continuous and its time derivative along the trajectories of~\eqref{eq:desired_formation_dyn} is:
		\be
		\dot{V}=\begin{cases}
			-C_d\norm{\mathbf{v}_{df}}^2 & \text{ if } \norm{\bar{\mathbf{r}}_{df}}<\frac{u_{md}^h}{k_1},\\
			-C_d\norm{\mathbf{v}_{df}}^2 & \text{ otherwise}.
		\end{cases}
		\ee
		$\dot{V}$ is negative semi-definite and we have from the dynamics \eqref{eq:desired_formation_dyn} that the largest invariant subset in $\calQ=\{\bar{\mathbf{r}}_{df},\mathbf{v}_{df} \in \bR^2| \dot{V}=0\}$ is the origin $\bar{\mathbf{r}}_{df}=\mathbf{v}_{df}=\mathbf{0}$. Using Lasalle's Invariance Principle (Theorem 4.4 in \cite{khalil2015nonlinear}), the trajectories of the system \eqref{eq:desired_formation_dyn} converge to $\bar{\mathbf{r}}_{df}=\mathbf{v}_{df}=\mathbf{0}$, i.e, the center $\mathbf{r}_{df}$ converges to $\mathbf{r}_s$ and so does the desired formation $\mathscr{F}_d^h$.
		From Theorem \ref{thm:stringNet_formation}, the defenders track these desired trajectories under appropriately modified $\eqref{eq:def_control1}$ and hence herd the attackers to $\calS$.		
	\end{proof}
	\section{Simulations}\label{sec:simulations} 
	%
	We provide a simulation of 5 defenders herding an adversarial swarm of 4 attackers to $\calS$ with saturated control inputs whose theoretical analysis is currently an ongoing work. The trajectories of all the agents are shown in Fig.~\ref{fig:herdMultiAtt}. As observed, starting from the given initial conditions, the defenders are able to gather before the attackers reach close to $\calP$, form the StringNet around the attackers and herd them to $\calS$. The safety is assessed in terms of critical distance ratios:
	\bewn
	\arraycolsep=1.4pt
	\baa
	\Delta_{d}^{d} =&  \displaystyle \max_{j\neq j' \in I_d} \frac{\hat{R}_{dj}^{dj'}}{R_{dj'}^{dj}}, \Delta_{a}^{d} =\displaystyle \max_{i \in I_a,j \in I_d} \frac{\hat{R}_{dj}^{ai}}{R_{dj}^{ai}},\Delta_{a}^{a} = \displaystyle \max_{i\neq i' \in I_a} \frac{\hat{R}_{ai}^{ai'}}{R_{ai}^{ai'}}
	\eaa
	\eewn
	
	\bewn 
	\Delta_{a}^{o} = \displaystyle \max_{i \in I_a} \max_{k \in \calN_{ai}^o}\frac{\xi_{ok}^{m}}{ E_{ok}^{ai}} , \quad 
	\Delta_{d}^{o} = \displaystyle \max_{j \in I_d} \max_{k \in \calN_{dj}^o} \frac{\xi_{ok}^{m}}{ E_{ok}^{dj}},
	\eewn
	where $E_{ok}^{ai}, E_{ok}^{dj}$ are super-elliptic distances from $\calO_k$ defined as per expression in~\eqref{eq:superEllipticDist}.
	These ratios have to be less than 1 for no collisions. As observed from Fig.~\ref{fig:criticalRelDistances} all these ratios are less than 1 for all times ensuring no collisions.
	\vspace{-5pt}
	\begin{figure}[h]
		\centering
		\vspace{-10pt}
		\setlength{\abovecaptionskip}{2pt plus 3pt minus 2pt}
		\setlength{\belowcaptionskip}{-20pt plus 3pt minus 2pt}	
		\includegraphics[width=.95\linewidth,trim={7cm 1cm 7cm 1cm},clip]{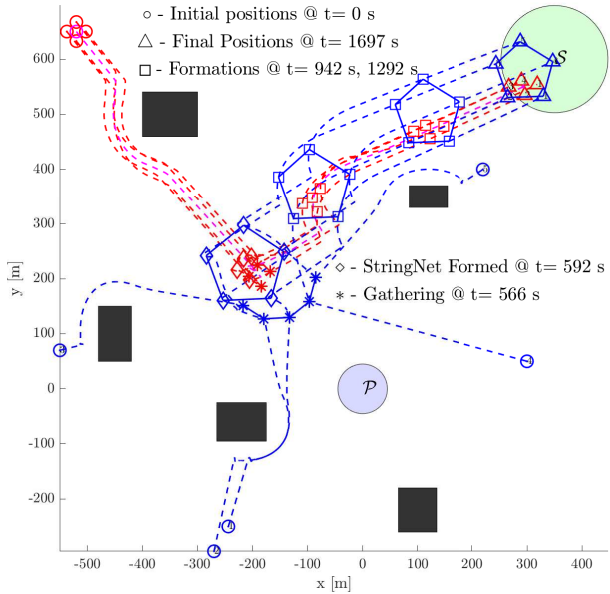}
		\caption{The herding paths.}
		\label{fig:herdMultiAtt}
	\end{figure}	
	\begin{figure}[h]
		\centering
		\setlength{\abovecaptionskip}{2pt plus 3pt minus 2pt}
		\setlength{\belowcaptionskip}{-10pt plus 3pt minus 2pt}
		\includegraphics[width=.85\linewidth,trim={.6cm .1cm 1cm .15cm},clip]{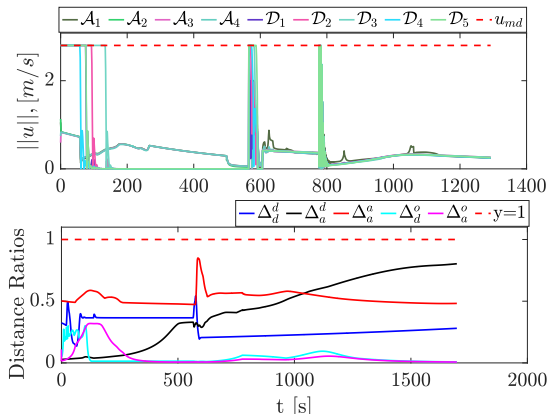}
		\caption{Inputs and critical distances.}
		\label{fig:criticalRelDistances}
	\end{figure}	
	
	\vspace{-1mm}
	\section{Conclusions and Ongoing Work} \label{sec:conclusions}
	We proposed a herding method for defending a protected area against an adversarial swarm. A closed formation is formed by the defenders around the attackers, restricts their motion and herds them to the safe area while avoiding the static rectangular obstacles.
	We provided formal analysis for the proposed approach and simulations with saturated control actions whose theoretical analysis and experimental investigation is a part of an ongoing work.  
	


	
	
	\vspace{-2mm}
	\bibliographystyle{IEEEtran}
	\bibliography{CDC2019_Refs}
\end{document}